\documentclass[aps,reprint,10pt,twocolumn,superscriptaddress]{revtex4-2}
\usepackage{amsmath, amssymb, amsthm, mathtools, bbm, physics}

\usepackage[utf8]{inputenc}
\usepackage{graphicx}
\usepackage{framed}
\usepackage{hyperref}
\usepackage{booktabs}

\usepackage{subcaption}
\usepackage{ragged2e} 
\DeclareCaptionJustification{justified}{\justifying}

\usepackage{diagbox}

\usepackage{microtype}
\microtypecontext{spacing=nonfrench}
\microtypesetup{
protrusion={true,nocompatibility},
activate={true,nocompatibility},
tracking=true,
kerning=true,
spacing={true}
}

\usepackage{tikz} 
\usetikzlibrary{calc, arrows, patterns, shapes, decorations, arrows.meta, fit, positioning}
\tikzset{
    auto,node distance =1 cm and 1 cm,semithick,
    var/.style ={circle, draw, minimum width = 1cm, ultra thick},
    latent/.style ={regular polygon, regular polygon sides=3, inner sep=1pt, draw, minimum width = 1.2cm, ultra thick},
    point/.style = {circle, draw, inner sep=0.06cm, fill, node contents={}},
    triangle/.style = {regular polygon, regular polygon sides=3, draw, inner sep=0.06cm, fill, node contents={}},
    bidir/.style={Latex-Latex,dashed},
    dir/.style={-Latex, thick},
    el/.style = {inner sep=2pt, align=left, sloped}
}
\tikzstyle{vertex}=[circle, fill=black!10, draw=black]
\tikzstyle{edge}=[thick]
\tikzstyle{clique}=[line width=4, draw=black!70]

\graphicspath{{imgs/}}

\newtheorem{prop}{Proposition}
\newtheorem{lemma}{Lemma}

\newcommand{\Do}[1]{\mathrm{do}(#1)}
\newcommand{\ACE}[2]{\mathrm{ACE}_{#1 \rightarrow #2}}
\newcommand{\qACE}[2]{\mathrm{qACE}_{#1 \rightarrow #2}}
\newcommand{\eye}{\mathbbm{1}}

\newcommand{\corr}[1]{\langle #1 \rangle}

\newcommand{\dcup}{\sqcup}
\newcommand{\pa}{\mathrm{pa}}
\newcommand{\ch}{\mathrm{ch}}
\newcommand{\Cl}{\mathcal{C}}
\newcommand{\io}{\mathrm{io}}
\newcommand{\mI}{\mathcal{I}}
\renewcommand{\vec}{\mathbf}

\newcommand{\asmobs}{\sigma_{a|x}}
\newcommand{\asmint}{\sigma_{\textrm{do}(a)}}

\begin{document}
\title{Observational-Interventional Bell Inequalities}

\author{Davide Poderini}
\affiliation{International Institute of Physics, Federal University of Rio Grande do Norte, 59078-970, Natal, Brazil}

\author{Ranieri Nery}
\affiliation{ICFO - Institut de Ciencies Fotoniques, The Barcelona Institute of Science and Technology, Castelldefels (Barcelona) 08860, Spain}

\author{George Moreno}
\affiliation{Departamento de Computação, Universidade Federal Rural de Pernambuco, 52171-900, Recife, Pernambuco, Brazil}

\author{Santiago Zamora}
\affiliation{International Institute of Physics, Federal University of Rio Grande do Norte, 59078-970, Natal, Brazil}
\affiliation{Departamento de F\'isica Te\'orica e Experimental, Federal University of Rio Grande do Norte, 59078-970 Natal, Brazil}

\author{Pedro Lauand}
\affiliation{Instituto de Física “Gleb Wataghin”, Universidade Estadual de Campinas, 130830-859, Campinas, Brazil}

\author{Rafael Chaves}
\affiliation{International Institute of Physics, Federal University of Rio Grande do Norte, 59078-970, Natal, Brazil}
\affiliation{School of Science and Technology, Federal University of Rio Grande do Norte, Natal, Brazil}

\begin{abstract}
Generalizations of Bell's theorem, particularly within quantum networks, are now being analyzed through the causal inference lens. However, the exploration of interventions, a central concept in causality theory, remains significantly unexplored. In this work we give an initial step in this direction, by analyzing the instrumental scenario and proposing novel hybrid Bell inequalities integrating observational and interventional data. Focusing on binary outcomes with any number of inputs, we obtain the complete characterization of the observational-interventional polytope, equivalent to a Hardy-like Bell inequality albeit describing a distinct quantum experiment. To illustrate its applications, we show a significant enhancement regarding threshold detection efficiencies for quantum violations also showing the use of these hybrid approach in quantum steering scenarios.
\end{abstract}

\maketitle
 
\section{Introduction}

Quantum correlations, those exhibited between two or more quantum systems and that cannot be explained in classical terms, are the core of quantum theory and its applications for information processing. From those, the correlations emerging in Bell experiments \cite{brunner2014bell} involving observers separated by space-like distances, represent the strongest form of nonclassical behavior. What makes them particularly remarkable is their ability to be verified solely based on assumptions about the underlying causal structure of the experiment, without any need to delve into the inner workings of measurement or state preparation devices. This device-independent framework offers a significant advantage in quantum information processing, ensuring the security and reliability of quantum protocols even if the devices are not accurate or even trustworthy.

Building upon Bell's theorem, recent research has unveiled that integrating causal networks with independent sources of correlations and communication between parties can lead to new forms of nonclassical behavior \cite{tavakoli2022bell}. The bilocal causal structure \cite{branciard2010characterizing,branciard2012bilocal,carvacho2017experimental,andreoli2017experimental,saunders2017experimental,sun2019experimental} underlying entanglement swapping \cite{zukowski1993event}, for example, allows for the activation of non-classicality in measurement devices \cite{pozas2019bounding}, the demonstration of the necessity of complex numbers \cite{renou2021quantum} and even self-testing of quantum theory \cite{weilenmann2020self}. Meanwhile, the triangle network produces nonclassicality without the need for measurement choices \cite{fritz2012beyond,renou2019genuine,chaves2021causal,polino2023experimental} and refined notions for multipartite nonlocality \cite{coiteux2021no,suprano2022experimental}. In turn, scenarios with communication can be employed to understand the requirements for the classical simulation of quantum states \cite{toner2003communication,brask2017bell}, to exclude specific nonlocal hidden variable models \cite{groblacher2007experimental,ringbauer2016experimental}, be applied in communication complexity \cite{buhrman2010nonlocality,ho2022entanglement}, also offering a new and under-explored way of detecting non-classical behavior through interventions rather than pure observations of a quantum system \cite{agresti2022experimental,gachechiladze2020quantifying}. 

Interventions are an essential tool in causal inference, enabling us to determine the causal relationships between variables in a given process \cite{pearl2009causality,balke1997bounds}. Unlike passive observations, interventions involve locally changing the underlying causal structure of an experiment, such as erasing all external influences that a given variable might have and putting it under the exclusive control of an observer. Interestingly, considering the case of the instrumental causal structure \cite{bonet2013instrumentality,chaves2018quantum,poderini2020exclusivity}, classical bounds on causal influence among the involved variables can be violated even when no Bell-like violation is possible \cite{agresti2022experimental,gachechiladze2020quantifying,liu2022quantum,cao2021detection,miklin2022causal}. 

Through interventions, one can demonstrate the quantum behavior of a system that may appear classical at the observational level. Previous works, however, have been limited to violations of causal bounds that rely on a specific measure of causal influence called the average causal effect (ACE) \cite{pearl2009causality,gachechiladze2020quantifying,hutter2023quantifying}. That is, all the interventional data from the experiment is coarse-grained in a single number, the ACE. Here we propose a new approach that considers all available interventional and observational data in a given experiment also considering its connection with quantum steering \cite{uola2020quantum,nery2018quantum}. Our approach defines a geometrical object that we name the observational-interventional polytope, which is bounded by hybrid (observational-interventional) inequalities that subsume all Bell-like and causal bounds previously considered in the literature \cite{pearl1995testability,bonet2013instrumentality,Kedagni2020}. As we show, by deriving new hybrid inequalities and applying them in a number of cases, this approach allows us to better detect and characterize non-classical behavior. 

The paper is organized as follows. In Sec. \ref{sec:sec2} we introduce the instrumental scenario, interventions and the known causal bounds. In sec. \ref{sec:sec3} we provide a full characterization of certain cases in terms of hybrid observational-interventional Bell inequalities, show how they improve the known bounds on detection efficiencies and discuss their equivalence to standard Bell inequalities. In Sec. \ref{sec:sec4} we generalize the approach beyond the instrumental scenario, showing the connection between DAGs with interventions with exogenized DAGs without interventions. In Sec. \ref{sec:steering} we describe the use of interventions in quantum steering \cite{wiseman2007steering,cavalcanti2016quantum}. Finally, in Sec. \ref{sec:discussion} we discuss our results and possible directions for future research.

\section{Instrumental scenario and interventions}
\label{sec:sec2}
We start discussing the instrumental causal structure \cite{pearl1995testability,chaves2018quantum,van2019quantum}, a paradigmatic scenario in which interventions and the quantification of causal influences play a central role.

If two variables $A$ and $B$ are found to be correlated, that is $p(a,b)\neq p(a)p(b)$, a basic question is to understand whether such correlations are due to direct causal influences from $A$ to $B$, or due to some common cause described (classically) by a random variable $\Lambda$. The variable $\Lambda$ is a confounding factor, the source of the mantra in statistics that ``correlation does not imply causation". More precisely, unless we have empirical control over such factors, we cannot distinguish between causal models of the type $A \rightarrow B$ from $A \leftarrow \Lambda \rightarrow B$. In most cases, however, confounding factors have to be treated as latent variables, and in order to reveal cause-and-effect relations one typically has to rely on interventions. Differently from passive observations, an intervention locally changes the underlying causal relations, erasing all causes acting on the variable we intervene upon. For instance, intervening on the variable $A$, we erase any correlation between $A$ and $B$ mediated by $\Lambda$. If after the intervention, we still observe correlations between $A$ and $B$, those can now be assertively associated with the direct causal influence $A \rightarrow B$.

Interventions provide a natural way to quantify causality \cite{janzing2013quantifying}, a widely used measure being the average causal effect, defined as \cite{pearl2009causality}
\begin{equation}
\label{eq:ACE_def}
    \ACE{A}{B} =  \max_{a,a',b} \left| p(b|\Do{a})-p(b|\Do{a'}) \right| \, ,
\end{equation}
The do-probability $p(b\vert\Do{a})=\sum_{\lambda}p(b\vert a,\lambda)p(\lambda)$ is associated with the intervention and differs from the observational probability  $p(b\vert a)=\sum_{\lambda}p(b\vert a,\lambda)p(\lambda\vert a)$. Interestingly, by introducing an instrumental variable $X$, controlled by the experimenter, and satisfying two causal assumptions, it is possible to infer $\ACE{A}{B}$ without the actual need for interventions. The instrumental variable is assumed to be independent of the confounding factors, that is,  $p(x,\lambda)=p(x)p(\lambda)$. Moreover, the correlations between $X$ and $B$ are mediated by $A$, that is, while $X$ has a direct causal influence over $A$ it does not over B. To illustrate the power of an instrumental variable, assume $A$ and $B$ are linearly related as $b=\gamma \cdot a + \lambda$, where we can interpret $\gamma$ as the amount of causal influence. Since $X$ and $\Lambda$ are statistically independent, we have that $\gamma= Cov(X,B)/Cov(X,A)$, where $Cov(X,A)=\corr{XA}-\corr{X}\corr{A}$ is the covariance between $A$ and $X$, and similarly for $Cov(X,B)$. That is, without any information about $\Lambda$, simply looking at the correlations between the instrument and the variables $A$ and $B$, one can estimate $\gamma$. 

Interestingly, this estimation of causal influences can also be performed in a device-independent setting where the functional relation between the variables is unknown. First notice that the underlying causal structure in the instrumental scenario is represented by the directed acyclic graph (DAG) depicted in Fig.~\ref{fig:instrumental_dag}. From this DAG and the causal Markov condition it follows that any observational distribution compatible with the instrumental scenario can be (classically) decomposed as
\begin{equation}
\label{eq:instrumental_markov} p(a,b|x)=\sum_{\lambda}p(a|x,\lambda)p(b|a,\lambda)p(\lambda) \,.
\end{equation}

As shown in \cite{balke1997bounds}, the observational probability distribution $p(a,b\vert x)$ also imposes restrictions on the interventional distribution, bounds of the form $\ACE{A}{B} \geq f(p(a,b\vert x))$, where $f(\cdot)$ is a linear function of the observed distribution. That is, with the use of an instrumental variable we can infer the effect of interventions without the actual need of performing them. Bounds on the $\ACE{A}{B}$ are thus hybrid Bell inequalities, that differently from the paradigmatic case, combine both observational and interventional data.

To illustrate, in what follows, we discuss in detail a few concrete cases, always assuming that all three  observable variables have finite cardinality that we denote as $l=|X|$, $m=|A|$, $n=|B|$, and we write their distribution as $p(a,b|x) = P(A=a, B=b| X=x)$ with $a\in \{0,\ldots,m-1\}, b\in \{0,\ldots,n-1\}$ and $x\in \{0,\ldots,l-1\}$. 

We will focus on the case where $m=n=2$. It is known that if we restrict attention to observational data, already with $l=4$ (an instrument assuming four different values) we generate all classes of observational Bell-like inequalities for the instrumental scenario \cite{Kedagni2020}, the so-called instrumental inequalities given by \cite{pearl1995testability,bonet2013instrumentality,Kedagni2020}
\begin{eqnarray}
\label{eq:instrumental}
& &  \mathcal{I}_1=p(0,0|0)+p(0,1|1) - 1\leq 0, \\ \nonumber
& & \mathcal{I}_2=p(0,1|0)-p(0,1|1)-p(1,1|1)-p(1,0|2)-p(0,1|2)\leq 0, \\  \nonumber
& & \mathcal{I}_3=p(0,0|0)+p(1,0|0)-p(0,1|1)-p(1,0|1)-p(0,0|2)\\ \nonumber
& &-p(1,0|2)-p(0,0|3)-p(1,1|3)\leq 0,
\end{eqnarray}
where we have considered an example of each class of inequalities (obtained by relabelling of the variables \cite{miklin2022causal}).

When considering causal bounds of the form $\ACE{A}{B} \geq f(p(a,b\vert x))$, apart from the scenario with $l=2$, there is no complete or systematic characterization of the corresponding inequalities,
the known causal bounds inequalities being given by
\cite{balke1997bounds,miklin2022causal}
\begin{equation}
\label{eq:ACEbound} 
\ACE{A}{B} \geq C_i,
\end{equation}
with
\begin{eqnarray}
\label{eq:ace34}
C_1 & = &  2 p(0,0|0)+p(1,1|0 ) + p(0,1|1)+p(1,1|1)-2, \nonumber\\
     C_2 & = & p(0,0|0) + p(0,0|2) + p(1,0|0) + p(1,1|1) \nonumber \\ &+& p(1,1|2) - 2,\\
     C_3 & = & p(0,0|0)+p(0,0|1)-p(0,1|1)+p(0,1|2) \nonumber \\ \nonumber &+ &p(1,0|0)  - p(1,0|1)+p(1,1|1)+p(1,1|2)-2.
\end{eqnarray}

However, not only these bounds are incomplete but also they rely upon a single parameter, a coarse-graining over the full do-probability information contained in $p(b\vert do(a))$. As we will show next, considering the full data of the instrumental scenario, composed of the observational data $p(a,b|x)$ and the interventional data $p(b\vert do(a))$, we can obtain a complete and very concise description of the instrumental scenario.

\section{Observational-interventional Bell inequalities}
\label{sec:sec3}
The geometry of the problem in the case of $l$ settings for $X$ and dichotomic outcomes ($m=n=2$) simplifies considerably when we combine the observational and interventional data. For any $l$ the set of inequalities defining the allowed distributions, are given by (up to relabeling of the variables) to two classes of inequalities. The first is given by
\begin{equation}
     p(b|\Do{a}) - p(a,b|x) \ge 0,
\end{equation}
a trivial inequality in the sense that it supports no quantum violations.
The second class, that we call $I_{l22}$ inequality, is given by
\begin{align}
\nonumber
 I_{l22}  &= p(b|\Do{a}) -p(a,b|x') + p(a,\bar b|x) \\
     &+ p(\bar a,b|x) - p(\bar a,b|x') \ge 0,
     \label{eq:m22_prob_ineqs}
\end{align}
a non-trivial inequality that, as we will show in the following, can be violated quantum-mechanically. 

Since $A$ and $B$ are dichotomic we can equivalently describe this set in terms of the correlators (expectation values): $\corr{AB}_x = \sum_{ab} (-1)^{a+b} p(ab|x)$, $\corr{B}_x = \sum_{ab} (-1)^{b} p(ab|x)$ and $\corr{B}_{do(a)} = \sum_{b} (-1)^{b} p(b|\Do{a})$.
In this representation, the constraints translate to (up to relabeling of the variables)
\begin{align}
    \nonumber
    &|\corr{AB}_x - \corr{B}_x + 2\corr{B}_{\Do{1}}|-1 \le \corr{A}_x, \\
    &|\corr{AB}_x - \corr{B}_{x'} + \corr{B}_{\Do{1}}| \le 1
    \label{eq:m22_corr_ineqs}
\end{align}
for any $x\neq x' \in \{0,\ldots,l\}$.
\begin{prop}
    The set of allowed observable and interventional distribution in the $l22$ set is completely characterized by~\eqref{eq:m22_prob_ineqs} or, equivalently, by~\eqref{eq:m22_corr_ineqs}.
\end{prop}
\begin{proof}
    We assume that this is true for the case $l=2$, the proof of which can be found in \ref{subsec:hardy_ineq}.
    Consider now a distribution $p(ab|x)$ and $p(b|\Do{a})$, with $|X|=l$ settings that respects~\eqref{eq:m22_corr_ineqs}.
    Since the proposition is true for $l=2$, this means that, for each couple of settings $x_1 = i, x_2 = j$ the distribution $p$ restricted to them, is compatible with the $222$ case, hence there is a joint distribution $q_{ij}(a_{i}, a_{j}, b_{0}, b_{1})$ from which we can obtain $p$ by marginalizing appropriately: 
    \begin{align}
        \nonumber
        p(b|\Do{a}) = \sum_{a_{i}, a_{j}, b_{\bar{a}}} q_{ij}(a_{i}, a_{j}, b_{0}, b_{1})\\
        \nonumber
        p(ab|x=i) = \sum_{a_{j}, b_{\bar{a}}} q_{ij}(a_{i}=a, a_{j}, b_{0}, b_{1}) \; ,
    \end{align}
    and similarly for $x=j$.
    Since $p(b|\Do{a}) = q_{ij}(b_a)$ is fixed for any couple of $i,j$, this defines the family of marginals $q_{ij}(b_0, b_1)$ up to one parameter $s_{ij}$ as:
    \begin{align}
        \nonumber
        q_{ij}(0, 0) &= s_{ij} \\ \nonumber
        q_{ij}(0, 1) &= p(0|\Do{0}) - s_{ij} \\ \nonumber
        q_{ij}(1, 0) &= p(0|\Do{1}) - s_{ij} \\ \nonumber
        q_{ij}(1, 1) &= 1 + s_{ij} - p(0|\Do{0}) - p(0|\Do{1})
    \end{align}
    Likewise, the marginals $q_{ij}(a_i, b_0, b_1)$ are defined by the observational data up one parameter $t_{ij}$ as:
    \begin{align}
        \nonumber
        q_{ij}(0, 0, 0) &= t_{ij} \\ \nonumber
        q_{ij}(0, 0, 1) &= p(0,0|i) - t_{ij} \\ \nonumber
        q_{ij}(0, 1, 0) &= p(0|\Do{1}) - p(1,0|i) - t_{ij} \\ \nonumber
        q_{ij}(0, 1, 1) &= p(0,1|i) + p(1,0|i) + t_{ij} - p(0|\Do{1})
    \end{align}
    with the only restriction that $t_{ij}/s_{ij} = q_{ij}(a_i | b_0 b_1) \le 1$, i.e. $t_{ij} \le s_{ij}$.
    Hence we can choose the $q_{ij}$ so that they have the same marginal $q(b_0, b_1)$ by choosing an appropriate $s \ge t_{ij} \forall i,j$, and we can write them as $q_{ij}(a_i, b_0, b_1) = q_{ij}(a_i, | b_0, b_1) q(b_0, b_1)$.
    We can then construct a complete joint probability distribution $Q$ as follows:
    \begin{multline}
        Q(a_1,\ldots,a_l, b_0, b_1) = \\
        = q_{01}(a_0|b_0 b_1) q_{12}(a_1|b_0 b_1) \cdots q_{l0}(a_l|b_0 b_1) q(b_0, b_1)
    \end{multline}
    This joint distribution reproduces the required observable and interventional distributions, proving that they are compatible with the $l22$ scenario.
\end{proof}

\subsection{A quantum model for observations and interventions}

\begin{figure}
\begin{subfigure}{.49\textwidth}
\begin{tikzpicture}
    \node[var] (b) at (2,0) {$B$};
    \node[var] (a) at (0,0) {$A$};
    \node[var] (x) at (-1,1.5) {$X$};
    \node[latent] (l) at (1,1.5) {$\Lambda$};
    \path[dir] (x) edge (a) (a) edge (b); 
    \path[dir] (l) edge (a) (l) edge (b);
\end{tikzpicture}
\caption{The Instrumental scenario.}
\label{fig:instrumental_dag}
\end{subfigure}
\begin{subfigure}{.49\textwidth}
\centering
    \begin{tikzpicture}
        \node[var] (x) at (-1,1.5) {$X$};
        \node[var, thick, double] (a) at (0,0) {$a$};
        \node[var] (b) at (2,0) {$B$};
        \node[latent] (l) at (1,1.5) {$\Lambda$};
        \path[dir] (a) edge (b) (l) edge (b);
    \end{tikzpicture}
    \caption{The instrumental scenario after an intervention on $A$.}
\label{fig:intervention_dag}
\end{subfigure}
\end{figure}

The description above considers an underlying classical theory of cause and effect that is incompatible with quantum predictions when the latent source $\Lambda$ is given by an entangled quantum state.
In this case, the probability distribution $p(a,b|x)$ is described by the Born rule as
\begin{equation}
\label{eq:instrumental_quantum}
    p(a,b|x)= \tr[M_x^{(a)} \otimes N_a^{(b)} \rho] \;,
\end{equation}
where $\{M_x^{(a)}\}_a$ and $\{N_a^{(b)}\}_b$ are POVM measurement for $A$ and $B$ respectively, and $\rho$ is their shared quantum state. In turn, the quantum version of the do-distribution, for an intervention on $A$ is given by
\begin{equation}
\label{eq:q_do_a}
   p(b|\Do{a}) = \tr[\left(\eye \otimes N_a^{(b)}\right) \rho] \,,
\end{equation}
which in turn allows us to define the quantum ACE as
\begin{equation}
\label{eq:qACE}
    \qACE{A}{B}= \max_{a,a',b} \tr[\eye \otimes \left(N_a^{(b)}-N_{a'}^{(b)}\right) \rho] \,.
\end{equation}

It can be shown that \cite{gachechiladze2020quantifying,agresti2022experimental}
\begin{equation}
\label{eq:qACE_witness}
    \qACE{A}{B} < C_i \le \ACE{A}{B} \,,
\end{equation}
that is, $\qACE{A}{B}$ obtained by interventions on a quantum system can violate the classical causal bounds. Remarkably, this sort of non-classicality can be achieved in scenarios where no instrumental Bell inequality can be violated (for instance, with $l=2$~\cite{henson2014theory}). That is, the use of interventional data allows to reveal the quantum nature of an experiment that relying only upon the observational data would have a classical explanation. 

One can check that this is the case also for the general scenario with dichotomic measurement $l22$, considering the class of inequalities~\eqref{eq:m22_prob_ineqs}.
Using the definitions~\eqref{eq:instrumental_quantum} and~\eqref{eq:q_do_a}, any constraint of the form~\eqref{eq:m22_prob_ineqs} with $x \neq x'$ can be violated up to $I_{l22} = -(\sqrt{2}-1)/2 \approx -0.2071$ with a bipartite state of two qubits $\ket{\Phi^+} = (\ket{00} + \ket{11})/\sqrt{2}$ and projective measurements (see appendix~\ref{sec:I22_qstrategy}).
Moreover, with numerical optimization (using the Navascues-Pironio-Acin (NPA) hierarchy~\cite{navascues2008convergent}) we can confirm that this corresponds to the maximum quantum bound. It is interesting to notice that while the general hybrid inequality we derive here achieves its optimal quantum violation with a maximally entangled state, the causal bound previously considered reaches its maximal quantum violation of partially entangled two-qubit states \cite{gachechiladze2020quantifying}.

\subsection{Detection efficiencies for the violation of hybrid inequalities}

As happens with standard (observational only) Bell inequalities, low detection efficiency is one of the main obstacles to the violation of hybrid Bell inequalities. When focusing solely on coincidence counts, a hybrid inequality can be violated, suggesting non-classical behavior. However, this inference is contingent upon the exclusion of non-detection events. The complete dataset, encompassing both detections and non-detections, might still conform to a classical causal model unless an additional fair sampling assumption is included. While this may seem natural in experiments probing the foundations of quantum mechanics, the introduction of such extra assumptions runs counter to the tenets of a device-independent approach. This becomes especially problematic in cryptographic settings, where an eavesdropper exploiting the detection inefficiency could take advantage of it to hack a key distribution protocol without being detected.

In the following, we analyze the minimum detection efficiencies required for a quantum violation of the $I_{l22}$ inequality. Previous works \cite{cao2021detection} based on causal bounds of the form \eqref{eq:ACEbound} have shown that a quantum violation requires detection efficiencies above $92\%$. As we will show next, the hybrid inequality we propose here improves significantly on this critical efficiency. In fact, the $I_{l22}$ inequality
displays the same features of the CHSH inequality \cite{eberhard1993background,wilms2008local}, a fact that as we prove later on is based on the equivalence of $I_{l22}$ with a Hardy-like Bell inequality \cite{Hardy1993}.

There are different manners to model detection inefficiencies \cite{wilms2008local,chaves2011feasibility,branciard2011detection} and their applicability might depend on the particular experimental setup. Here we consider a model of particular relevance for photonic qubits where the observers need to distinguish two orthogonal polarizations. As the measurement apparatus has only one polarization-sensitive detector, the absence of a photon (the non-detection event) cannot be distinguished from a photon with the wrong polarization (an event we call $a^*$/$b^*$). 
\begin{figure}[t!]
\centering  
\includegraphics[clip,width=0.95\columnwidth]{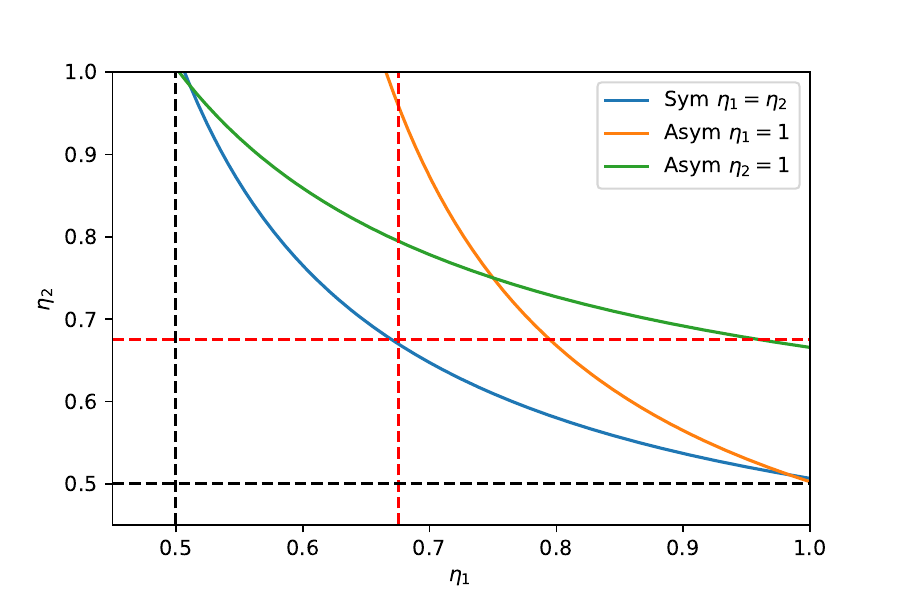} 
\caption{\textbf{Minimum detection inefficiencies for the violation of the $I_{l22}$ inequality.} In the symmetric case (represented by the blue curve) the minimum value of the efficiency found was $\eta_1=\eta_2=2/3$.  In the asymmetric case with one perfect detector, for instance, $\eta_1=1$, the minimum detection efficiency for the other $\eta_2=0.5$ (and vice-versa). These critical efficiencies are represented by the black dashed lines. The green and orange curves were obtained with the optimal parameters for $\eta_1=0.51$ and $\eta_2=0.51$, respectively.}
\label{fig:In222_eff}
\end{figure}

Employing the general framework proposed in \cite{wilms2008local} and the fact that \cite{gachechiladze2020quantifying}
\begin{equation}
    p(ab|x) = p^{Bell}(ab|xa) \label{InstP},
\end{equation}
while the interventional probabilities $p(b|do(a))$ are related as 
\begin{equation}
    p(b|do(a)) = \sum_{a'} p^{Bell}(a'b|xa),\label{DoP}
\end{equation}
where $p^{Bell}(ab|xy)$ are the observational probability distributions obtained in a standard Bell scenario, the ideal instrumental probabilities $p(ab|x)$ are related to the noisy probabilities $p_{\eta_1\eta_2}(ab|x)$ by 
\begin{align}
    &p_{\eta_1\eta_2}(ab|x) = \eta_1\eta_2p(ab|x) \nonumber, \\
    &p_{\eta_1\eta_2}(a^*b|x) = \eta_1\eta_2p(a^*b|x) + (1-\eta_1)\eta_2p(b|do(a^*)) \nonumber,\\
    &p_{\eta_1\eta_2}(ab^*|x) = \eta_1\eta_2p(ab^*|x) + \eta_1(1-\eta_2)p(a|x) ,\\
    &p_{\eta_1\eta_2}(a^*b^*|x) = \eta_1\eta_2p(a^*b^*|x) + \eta_1(1-\eta_2)p(a^*|x)  \nonumber\\
   &\hspace{2.5 cm}+ (1-\eta_1)\eta_2p(b^*|do(a^*)) \nonumber\\ 
   &\hspace{2.5 cm}+(1-\eta_1)(1-\eta_2).\nonumber
   \label{InefficienciesIns}
\end{align}

Within this setup, the $I_{l22}$ inequality was optimized in the case of identical detectors, the symmetric case $\eta_1=\eta_2$, and in the case of asymmetric detectors. For the symmetric case, the minimum $\eta_m$ was found to be $\eta^s_m\approx.6667 = 2/3$. In the asymmetric case, keeping the efficiency of one detector equal to $\eta=1$, the minimum detection efficiency for the other is $\eta^{as}_m=0.5$. In Fig. \ref{fig:In222_eff} we plot the regions for violation of the $I_{l22}$ inequality as a function of $\eta_1$ and $\eta_2$, optimizing over two-qubit states and projective measurements. The intersection of the dashed curves denotes the minimum value of $\eta_1$ and $\eta_2$ to have a violation of the inequality. The black dashed curve corresponds to the minimum value of $\eta$ in the asymmetric case, and the red dashed curve corresponds to the minimum value of $\eta$ in the symmetric case.

The region of violations displayed in Fig. \ref{fig:In222_eff} is identical to what one obtains for the CHSH inequality \cite{wilms2008local}. Indeed, it is known that observational instrumental inequalities can be mapped to standard Bell inequalities \cite{van2019quantum} via the mapping in \eqref{eq:q_do_a}. In particular, the inequality $\mathcal{I}_3$ in Eq. \eqref{eq:instrumental} can be mapped via a lifting \cite{pironio2014all} to the CHSH inequality. More precisely, inequality $\mathcal{I}_3$ can be seen as the CHSH inequality if a fourth measurement (thus $l= \vert X \vert =4 $) is introduced. In this sense, one should expect that the $l22$ scenario recovers the same features of the CHSH scenario. Notice, however, that by including not only observational but interventional data as well, we can recover the minimum detection efficiencies of the CHSH scenario already with $l= \vert X \vert =2 $, a clear advantage of interventions.

\subsection{Equivalence of the interventional inequality to a Hardy-like inequality}
\label{subsec:hardy_ineq}
Herein we prove that the inequality \eqref{eq:m22_prob_ineqs} completely characterizes the classically compatible data tables $p(ab|x)$ and $p(b|do(a))$ with the instrumental scenario for the case $l=2$. To achieve this, we show that the $I_{l22}$ inequalities for $l=2$ are equivalent to considering the Hardy-type Bell inequalities in the simplest Bell scenario. 

The ``Hardy's paradox" \cite{Hardy1993} can arguably be understood as an alternative formulation of Bell's theorem. Truly, for classical models in the Bell scenario, i.e. for 
\begin{equation}
    p(a,b|x,y)=\sum_{\lambda}p(a|x,\lambda)p(b|y,\lambda)p(\lambda),
\end{equation}
if we observe 
\begin{equation}
\label{eq:hardy_assump}
    p(1,0|0,1)=p(0,1|1,0)=p(0,0|0,0)=0,
\end{equation} 
then we must have $p(0,0|1,1)=0$. To see this, we first remark that we can, with no loss of generality, assume that $p(\lambda)$ has full support, i.e. $p(\lambda)>0$. Then, given \eqref{eq:hardy_assump} we obtain the equalities 
\begin{equation}
    \begin{aligned}
        &p(a=1|x=0,\lambda)p(b=0|y=1,\lambda)=0\\
        &p(a=0|x=1,\lambda)p(b=1|y=0,\lambda)=0\\
        &p(a=0|x=0,\lambda)p(b=0|y=0,\lambda)=0.\\
    \end{aligned}
\end{equation}
The first condition gives us two possibilities: either $p(b=0|y=1,\lambda)=0$, in which case $p(0,0|1,1)=0$ follows directly from the classical decomposition, or $p(a=1|x=0,\lambda)=0$, which implies $p(a=0|x=0,\lambda)=1\implies p(b=0|y=0,\lambda)=0\implies p(b=1|y=0,\lambda)=1 \implies p(a=0|x=1,\lambda)=0 \implies p(0,0|1,1)=0$ from the classical decomposition. However, quantum theory predicts one can reach $p_Q(1,0|0,1)=p_Q(0,1|1,0)=p_Q(0,0|0,0)=0$ and $p_Q(0,0|1,1)>0$ providing a direct contradiction of quantum versus classical predictions. 

At first, Bell inequalities and Hardy’s paradox may seem like fundamentally different concepts — Bell is concerned with expectation values, and Hardy with the logical possibility and impossibilities of events. However, due to experimental imperfections, events that ``never" happen are virtually certain to occur eventually. To deal with these conditions, we can generalize Hardy's argument by looking at it as a game between two parts, assigning a cost to a particular event $(a,b|x,y)$ \cite{Mancinska_2014}. The Hardy inequalities can capture this game reminiscent of Hardy's paradox, inequalities of the form
\begin{equation}
\label{eq: hardy_ineq}
    p(1,0|0,1)+p(0,1|1,0)+p(0,0|0,0)\geq p(0,0|1,1). 
\end{equation}
By taking into account the possible relabelings of the outcomes and measurement settings we have
\begin{equation}
p(\bar{a},b|x,y')+p(\bar{a},b|x',y)+p(a,\bar{b}|x,y)\geq p(\bar{a},b|x',y'),
\end{equation}
where $x\neq x'$, $y\neq y'$, $\bar{a}=a\oplus 1$ and similarly for $\bar{b}$. In particular, we can choose the case $y'=\bar{a}$, $y=a$ and use condition \eqref{DoP} and \eqref{InstP} to obtain the corresponding inequality in terms of observational data $p(a,b|x)$ and interventional data $p(b|\Do{a})$. Notice that, for the case where $|A|=2$, by combining \eqref{DoP} and \eqref{InstP} one has a bijective map between $p(a,b|x)$, $p(b|\Do{a})$ and $p(a,b|x,y)$, namely 
\begin{equation}
    \begin{aligned}
        &p(a,b|x,y=a)=p(a,b|x)\\
        &p(\bar{a},b|x,y=a)=p(b|\Do{a})-p(a,b|x).\\
    \end{aligned}
    \label{eq:instr2bell}
\end{equation}
This procedure yields 
\begin{equation}
p(\bar{a},b|x)+p(b|\Do{a})- p(a,b|x')+p(a,\bar{b}|x)\geq p(\bar{a},b|x'),
\end{equation}
which is precisely what we have in Eq. \eqref{eq:m22_prob_ineqs}. Finally, to prove these inequalities completely characterize the interventional polytope, i.e. the polytope delimited by the observational and interventional data tables, we show in Appendix \ref{app:Hardy_CHSH} how from the Hardy-type inequality \eqref{eq: hardy_ineq} (and its relabelings) we can obtain the CHSH inequalities. We prove this fact for completeness. However, the connection between CHSH and the Hardy inequalities has been studied at length, and for details, we refer the reader to \cite{Mansfield_Fritz_2012, Yang_2019, GHIRARDI20081982}. 

We reinforce that even though our hybrid inequality is equivalent to a Hardy-like inequality and thus to the CHSH inequality, this does not mean that the scenario we consider here is just Bell in disguise. The underlying causal structure in an instrumental and Bell scenario are fundamentally different from a physical perspective. In the Bell case we have a spatial scenario, where the correlations between the space-like separated parties are mediated by the source, the shared quantum state between them. In turn, for the instrumental case, we have a time-like scenario, where one of the parties has a direct causal influence over the other, reason why interventions add more structure and information to the problem of detecting non-classicality.

\section{Generalized observational-interventional framework}
\label{sec:sec4}

As we mentioned, the applicability of interventional methods is not limited to the previous examples, but can be used to analyze a vast range of causal scenarios. In this section we will present a general framework that can in principle be used, together with other techniques, to study arbitrary structures when interventional data is available.
Interestingly, using this formalism we will see how the result of section~\ref{subsec:hardy_ineq} is not accidental, but a feature of a more general property.

Let us consider a causal scenario described by a DAG $G=(V, E)$, where $V = U \dcup W$, $U$ and $W$ being the sets of observable and latent nodes respectively, and $E$ the set of edges (directed arrows) between them.
For each node $n \in U$ we denote the set of parent nodes as $\pa_G(n)$ or simply $\pa(n)$ if there is no ambiguity. 
A distribution $p(x_U) = p(\{x_n\}_{n \in U})$ for the variables $X_n$ is said to be \emph{compatible} with $G$ if it obeys the Markov factorization condition:
\begin{equation}
    p(x_U) = \sum_{\vec \lambda} \prod_{n \in V} p(x_n| x_{\pa(n)})\, ,
\end{equation}
where $\vec \lambda$ ranges over all the possible values of latent nodes in $W$.
We denote the set of such distributions as $\Cl(G)$.

Assuming that each observable variable $X_n$ has finite cardinality, it is known that any compatible distribution has a description in terms of finite dimensional latent variables~\cite{rosset2017universal, fraser2020combinatorial}. 
This means that we can equivalently describe the process with a set of random variables $\{\Lambda_l: \Omega \rightarrow \Omega_l]\}_{l \in  W}$, complemented, if necessary, by a set of local noise variables $\{N_k: \Omega \rightarrow \Omega_k\}_{k \in U}$, on some finite dimensional sample set $\Omega$, and a collection of functions $\{f_k\}_{k \in U}$ defining the observable variables as:
\begin{equation}
    X_k = f_k(X_{\pa(k)}, N_k) \, .
    \label{eq:struct_eq}
\end{equation}
This description is called a \emph{Functional Causal Model} and \eqref{eq:struct_eq} are usually called \emph{Structural equations}~\cite{pearl2009causality}.

This description makes the effect of interventions more transparent.
Let us say that we want to perform a sequence of interventions in different subsets $\mI = \{I_m \subseteq U\}_n$ of the observable nodes.
We can identify this as the interventional scenario described by the couple $(G, \mI)$.
The \emph{intervention}, for each set $I \in \mI$, effectively corresponds to considering a different collection of random variables, which we denote as $X_k^{\Do{I}}$, described by the structural equations
\begin{equation}
    X_k^{\Do{I}} = f_k(X_{\pa(k) \setminus I}, x_I, N_k)\, ,
\end{equation}
sharing the same functional relationships as before, but where now the values of the nodes belonging to $I$ is fixed.
From this, we can easily define the \emph{observational-interventional distribution vector} as $\vec p_{\io} = \vec p_0 \oplus \bigoplus_n \vec p_{n}$, where $\vec p_0$ represent the vector associated with the observable distribution $p(x_{U})$, while $\vec p_n$ the one corresponding to the interventional one, for each $I_n$, that is $p\left(x_{U}^{\Do{I_n}}\right)$.
Similarly to the observable-only case, we say that such $\vec p_{\io}$ is compatible with the interventional scenario $(G, \mI)$, formally $\vec p_{\io} \in \Cl_\mI(G)$.

\subsection{Relationship with the exogenized graph}
The appearance of Hardy's inequality for the Bell scenario in the interventional version of the Instrumental scenario, reveals an interesting relationship between the two, that can be easily generalized.

\begin{figure}
\begin{subfigure}{.5\textwidth} 
\begin{tikzpicture}
    \node[var] (b) at (2,0) {$B$};
    \node[var] (c) at (4,0) {$C$};
    \node[var] (a) at (0,0) {$A$};
    \node[latent] (l) at (1,1.5) {$\Lambda$};
    \node[latent] (n) at (3,1.5) {$N$};
    \path[dir] (a) edge (b) (b) edge (c); 
    \path[dir] (l) edge (a) (l) edge (b);
    \path[dir] (n) edge (c) (n) edge (b);
\end{tikzpicture}
\caption{The Interventional scenario.}
\label{fig:exoint_intdag}
\end{subfigure}
\begin{subfigure}{.5\textwidth}
\begin{tikzpicture}
    \node[var] (b) at (1.5,0) {$B$};
    \node[var] (bb) at (3,0) {$\bar{B}$};
    \node[var] (c) at (4.5,0) {$C$};
    \node[var] (a) at (0,0) {$A$};
    \node[latent] (l) at (1,1.5) {$\Lambda$};
    \node[latent] (n) at (3,1.5) {$N$};
    \path[dir] (a) edge (b) (bb) edge (c); 
    \path[dir] (l) edge (a) (l) edge (b);
    \path[dir] (n) edge (c) (n) edge (b);
\end{tikzpicture}
\caption{The corresponding exogenized scenario.}
\label{fig:exoint_exodag}
\end{subfigure}
\caption{The interventional scenario and its corresponding exogenized graph.
In this example the intervention set $I = \{B\}$ is composed by only one element.
The procedure splits the $B$ node in $B$, with  only incoming edges, and $\bar B$, having instead only outgoing edges.}
\label{fig:exoint_dags}
\end{figure}
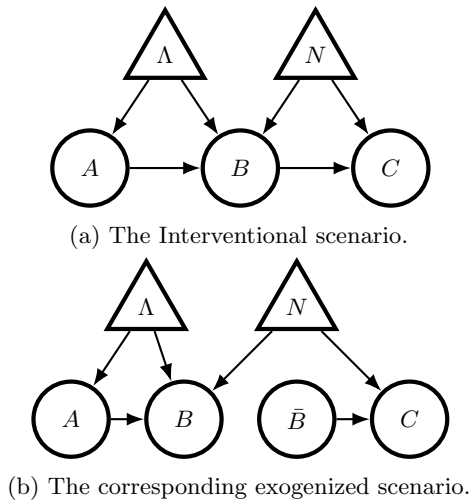

Like before, consider a DAG $G = (V, E), V = U \dcup W$ and a subset $I \subset U$ of its observable nodes, let us define the \emph{exogenized} DAG $G_I$ as the graph where we
\begin{enumerate}
    \item Include additional nodes and variables $X_{\bar{i}}$ for each $i \in I$.
    \item Add edges $(\bar i, n)$ for each $n \in \ch(i)$.
    \item Remove all incoming edges from the nodes $i \in I$.
\end{enumerate}
In other words, $G_I$ can be regarded as the graph where each node belonging to $i \in I$ is split in two nodes $i$ and $\bar i$, the first retaining only the incoming edges, and the second having only the outgoing ones of the original node (see the example in Fig.~\ref{fig:exoint_dags}).

\begin{prop}
    Consider a DAG $G = (V,E)$ and a corresponding interventional scenario $(G, \{I\})$.
    There always exists a surjective mapping $g: \Cl(G_I) \rightarrow \Cl_I(G)$      
\end{prop}
\begin{proof}
    Consider a certain distribution $q \in \Cl(G_I)$.
    For each random variable we can then write structure equations $X_n = f_n(X_{\pa(n) \setminus I}, X_{\bar\pa(n)}, N_n)$.
    From this description it follows immediately that the variables in the interventional and observational case can be obtained by fixing the values of $X_{\bar{i}}$ respectively to i) an arbitrary value $x_{\bar{i}}$ and ii) the same value of the corresponding node in the original $G$, $X_{\bar{i}} = X_i$.
    We can then define the function $g$ accordingly as the one that takes any $\vec q \in \Cl(G_I)$ to $g(\vec q) = (\vec p_0, \vec p_I)$, where
    \begin{align}
        \vec p_0 &= q(x_U | \{ x_{\bar{i}} = x_i \}_{i \in I}) \\
        \vec p_I &= \sum_{x_I} q(x_U | \{ x_{\bar{i}} \}_{i \in I})\, .       
    \end{align}

    To show that $g$ is surjective we can notice that, given a $\vec p \in \Cl_I(G)$, one can always construct a $\vec q \in \Cl(G_I)$ that maps to $\vec p$ under $g$, by again describing the system using structure equations, and getting the distribution generated by the same functions $f_k$ but substituting $X_i$ for $X_{\bar{i}}$ for any $i \in I$ whenever it appears in its argument.
\end{proof}

In this sense, the constraints imposed by a certain interventional scenario $(G, I)$, can be studied in the corresponding larger exogenized scenario $G_I$.
In general the interventional scenario can be less restrictive, in the sense that we could have some $q \not\in \Cl(G_I)$ for which $g(q) \in \Cl_I(G)$.
Nonetheless, in such cases as the instrumental scenario with dichotomic variables, whose corresponding exogenized graph is precisely the Bell DAG, the interventional constraints turn out to coincide, since the $g$ is also injective as shown by Eq.~\eqref{eq:instr2bell}.

\section{Interventions and Quantum Steering}
\label{sec:steering}

Moving beyond device-independent scenarios, the DAG representation of causal models can be effectively adapted to the use of quantum observable nodes \cite{barrett2019quantum}, which then can be interpreted as channels from input Hilbert spaces to output Hilbert spaces. Important quantum information protocols such as remote state preparation \cite{bennett2001remote} and dense coding \cite{bennett1992communication,moreno2021semi} or relevant measures of quantum correlations such as quantum steering \cite{wiseman2007steering,cavalcanti2016quantum} can be modeled in that manner. To illustrate, take the DAG represented in Fig. \ref{fig:instrumental_dag}, for example, and consider $B$ as a quantum node. This scenario is a direct adaptation from the instrumental scenario to the quantum case, but can also be seen as the scenario of quantum steering with communication allowed from the party with a black-box device to the party with a fully characterized device \cite{nery2018quantum}. The result of the channel associated with node $B$ is a set of local states that depend on the variable $A$ as a preparation parameter and the latent factor $\Lambda$ either as an additional parameter (if it is classical) or as the input state upon which a completely positive and trace preserving (CPTP) channel (that depends on $A$) acts.

The main object associated with such scenarios is called assemblage \cite{cavalcanti2016quantum}, a collection of states $\sigma_{a|x}$ with less-than-unit trace, that encodes both the conditional state obtained on node $B$ and the distribution associated with the values of $A$ and $X$: $p(a|x) = \tr[\sigma_{a|x}]$, $\rho_{a,x} = \sigma_{a|x}/p(a|x)$. In the case of a classical latent variable $\Lambda$, the causal Markov condition that leads to Eq. \eqref{eq:instrumental_markov} changes to
\begin{equation}
\label{eq:steerclas}
    \sigma_{a|x} = \sum_\lambda p(\lambda) p(a|x,\lambda) \rho^B_{a,\lambda},
\end{equation}
while in the case of a quantum $\Lambda$, we have that
\begin{equation}
\label{eq:steerq}
    \sigma_{a|x} = \tr\left[ M^{(a)}_x \otimes \mathcal{E}_a (\rho_\Lambda)\right],
\end{equation}
where the channel mapping a quantum state into a random variable is a measurement, represented here by the POVMs $M^{(a)}_x$, and the quantum-to-quantum channel is represented by CPTP maps $\mathcal{E}_a$ that act only on the subsystem of $\rho_\Lambda$ that is sent to node $B$.

Interventions can be described analogously, by eliminating the parents from the intervened node and observing the emergent causal Markov condition for the modified DAG. In the case of interventions affecting the edges of quantum nodes, either the output space is traced out if it is an output edge that is removed, or the identity is considered as an input if an input edge is removed. For instance, assemblages compatible with the intervened DAG of Fig. \ref{fig:intervention_dag} are given by
\begin{align}
    \sigma^C_{\mathrm{do(a)}} &= \sum_\lambda p(\lambda)\,\rho^B_{a,\lambda},     \label{eq:assemb_class_instrumental}\\
    \textrm{or}\,\,\sigma^Q_{\mathrm{do(a)}} &= \tr_{\Lambda_A}\left[ I \otimes \mathcal{E}_a(\rho_\Lambda) \right],
\end{align}
for classical and quantum $\Lambda$, respectively. We name the full collection of states, combining observational and interventional data, as the extended assemblage, and we represent it by $\{\sigma_{a|x}\}_{a,x} \oplus \{\sigma_{\textrm{do}(a)}\}_a$, where $\{\asmobs\}_{a,x}$ corresponds to the standard collection of states, called the observational part, and $\{\asmint\}_a$ corresponds to the interventional part.

It is known that considering only observational data, the classical and quantum descriptions, respectively, \eqref{eq:steerclas} and \eqref{eq:steerq}, are incompatible for scenarios with $|X|\geq3$, meaning that quantum correlations can be probed even in the presence of classical communication (and no inputs on the side with a trusted device) \cite{nery2018quantum}. For $|X|=2$, however, no violation of the classical description has been found so far and it is known that, with a classical $B$, no violation is possible if one is restricted to observational data only \cite{henson2014theory}. However, as we show next, non-classical behavior becomes possible already with $|X|=2$ if we consider the extended assemblage including interventional data, thus generalizing the device-independent results of \cite{chaves2018quantum,gachechiladze2020quantifying}.

A direct characterization of the set of permissible extended assemblages is difficult in general since the set boundary is characterized by a continuum of states. However, as proven in the next lemma, the sets for both classical and quantum correlations are each convex, which allows their characterization in terms of a convex optimization problem and ultimately leads to a semi-definite program for testing the compatibility for given assemblages. The lemma states that the interventional data part of the assemblage respects the convex combination that produced the original observational data assemblages. More precisely:
\begin{lemma}
Given two extended assemblages $\{\sigma^A_{a|x}\} \oplus \{\sigma^A_{\mathrm{do}(a)}\}$ and $\{\sigma^B_{a|x}\} \oplus \{\sigma^B_{\mathrm{do}(a)}\}$, compatibility with the DAG of Fig. \ref{fig:instrumental_dag} implies compatibility for the combined extended assemblage $\{p\,\sigma^A_{a|x} + (1-p)\,\sigma^B_{a|x}\} \oplus \{p\,\sigma^A_{\mathrm{do}(a)} + (1-p)\,\sigma^B_{\mathrm{do}(a)}\}$, for any $p \in [0,1]$. Moreover, this convexity holds for both classical and quantum latent $\Lambda$.
\end{lemma}

\begin{proof}
The proof resides in the fact that there is no limitation for the ``size'' of $\Lambda$, that is, on its cardinality in the classical case, or its dimension, in the quantum case. A convex combination of the respective decompositions can be made possible with the addition of a flag $\lambda_f \in \{A,B\}$, such that when the flag assumes the value $A$ with probability $p$, or $B$ with probability $1-p$, the corresponding decomposition for assemblages $\sigma^A_{a|x}$, $\sigma^A_{\mathrm{do}(a)}$, or $\sigma^B_{a|x}$, $\sigma^B_{\mathrm{do}(a)}$ should be used, in that order.

For the classical case, this is done explicitly by setting $\Lambda' = \Lambda \times \Lambda_f$ and 
\begin{equation}
    p_{\lambda'} = 
    \begin{cases}
    p\,p_\lambda,\quad &\textrm{for}\,\, \lambda' = (\lambda,A) \cr 
    (1-p)\,p_\lambda,\quad & \textrm{for}\,\, \lambda' = (\lambda,B)
    \end{cases},
\end{equation}
with the local distributions and local hidden states following a similar pattern: e.g., $p(a|x,\lambda') = p_A(a|x,\lambda)\,\delta_{\lambda_f,0} + p_A(a|x,\lambda)\,\delta_{\lambda_f,1}$, where $\delta$ is the Kronecker delta. 

For the quantum case, a similar strategy can be used by augmenting the shared state with flags, as $\rho_{\Lambda'} = p\,|00\rangle\langle00|_{F_A,F_B}\otimes\rho^A_\Lambda + (1-p)\,|11\rangle\langle 11|_{F_A,F_B}\otimes\rho^B_\Lambda$, where flag $F_A$ is sent to the A-node side and $F_B$, to B-node side. Local measurements and channels have to be implemented conditionally on projecting the corresponding flag onto $|0\rangle$ or $|1\rangle$. Measurements for the assemblage $\sigma^A$ are implemented conditioned on the flag being $|0\rangle$ and for $\sigma^B$ if the flag is $|1\rangle$. Similarly, the channels $\mathcal{E}^A_a$ and $\mathcal{E}^B_b$ are implemented if $F_B$ is $|0\rangle$ or $|1\rangle$, respectively.
\end{proof}


Once it is established that the set of classical extended assemblages is convex, a robustness-like measure for incompatibility with a classical description can be obtained in the form of a semi-definite program (SDP). Given an extended assemblage $\{\sigma_{a|x}\}_{a,x} \oplus \{\sigma_{\mathrm{do}(a)}\}_a$, it is either the case that it admits a classical decomposition, or that some classically correlated assemblage can be combined with it, in the form of noise, that results in a classically decomposable assemblage. The task is to find the minimal weight $p$, such that 
\begin{align}
    p\,\sigma^{\textrm{N}}_{a|x} + (1-p)\,\sigma_{a|x} &= \sigma^{\textrm{L}}_{a|x}, \nonumber \\
    p\,\sigma^{\textrm{N}}_{\mathrm{do}(a)} + (1-p)\,\sigma_{\mathrm{do}(a)} &= \sigma^{\textrm{L}}_{\mathrm{do}(a)},
\end{align}
where $\sigma^{\textrm{N}}_{a|x}$ and $\sigma^{\textrm{L}}_{a|x}$ admit the decomposition of Eq.\ \eqref{eq:steerclas} and $\sigma^{\textrm{N}}_{\mathrm{do}(a)}$ and $\sigma^{\textrm{L}}_{\mathrm{do}(a)}$ admit the decomposition of Eq.\ \eqref{eq:assemb_class_instrumental}. In the form presented above, the constraints would be nonlinear in the variables to be optimized and thus would not constitute an SDP--note that, in this formulation, we would need to optimize separately, for instance, both the noise term $\sigma^{\textrm{N}}_{a|x}$ and the weight $p$. 
To turn the problem into a valid SDP, some adaptations are required. An equivalent way of stating the problem is presented below, only involving linear equality and linear inequality constraints, which are then allowed to be solved as an SDP:
\begin{subequations}
\begin{align}
\min&\quad \tau \label{s_eq:prim_objective} \\
\textrm{s.t.}&\quad \sigma_{a|x} + \tilde{\sigma}^{\textrm{N}}_{a|x} = \tilde{\sigma}^{\textrm{L}}_{a|x} \label{s_eq:prim_noise_add} \\
&\quad \sigma_{\textrm{do}(a)} + \sum_\lambda \zeta_{a,\lambda} = \sum_\lambda \xi_{a,\lambda}, \label{s_eq:prim_interv_noise} \\
&\quad \tilde{\sigma}^{\textrm{N}}_{a|x} = \sum_\lambda D_\lambda(a|x)\,\zeta_{a,\lambda}, \label{s_eq:prim_local_state_noise} \\
&\quad \tilde{\sigma}^{\textrm{L}}_{a|x} = \sum_\lambda D_\lambda(a|x)\,\xi_{a,\lambda}, \label{s_eq:prim_local_state_result} \\
&\quad \zeta_{a,\lambda} \succeq 0,\,\, \mathrm{Tr}[\zeta_{a,\lambda}] = f_\lambda, \vphantom{\sum_\lambda} \label{s_eq:prim_noise_prob_lam} \\
&\quad \xi_{a,\lambda} \succeq 0,\,\, \mathrm{Tr}[\xi_{a,\lambda}] = g_\lambda, \vphantom{\sum_\lambda} \label{s_eq:prim_result_prob_lam} \\
\textrm{and}&\quad \sum_\lambda f_\lambda = \tau, \label{s_eq:robustness_normalization}.
\end{align}
\label{eq:Steering_Primal}
\end{subequations}
The adaptations involved correspond to optimizing over $\tau = p/(1-p)$, which, for $p \geq 0$ is a monotonically increasing quantity with respect to $p$ in the interval $[0,+\infty)$; embedding the value of $\tau$ into the overall trace of $\tilde{\sigma}^{\textrm{N}}_{a|x}$, so that constraints \eqref{s_eq:prim_local_state_noise} and \eqref{s_eq:robustness_normalization} combined imply $\sum_\lambda \tr[\zeta_{a,\lambda}] = \tau$ (or, equivalently, $\sum_a \tr[\tilde{\sigma}^{\textrm{N}}_{a|x}] = \tau)$. Also $\tilde{\sigma}^{\textrm{L}}_{a|x}$ is left unnormalized since, to match equality \eqref{s_eq:prim_noise_add}, it must hold that $\sum_\lambda \tr[\xi_{a,\lambda}] = 1 + \tau$, assuming the use of a normalized input assemblage $\asmobs$. Another manipulation of the original problem involves concentrating all probabilistic/non-deterministic behavior of $p(a|x,\lambda)$ into the probability of $\lambda$ itself, leaving the deterministic response functions $D_\lambda(a|x) = \delta_{a,f_\lambda(x)}$, where $\lambda$ now is an index for the functions that map $X$ to $A$ and $\delta_{a,a'}$ is the Kronecker delta. This way, the probabilities are combined into the weight of the local state terms, $\zeta_{a,\lambda}$ and $\xi_{a,\lambda}$, as, e.g., $\tr[\zeta_{a,\lambda}] = \tau \cdot p_\lambda$. Positivity of the local states ensures that a valid probability distribution can be obtained for $\Lambda$. 


One advantage of formulating the compatibility problem as an SDP is that it is immediately possible to obtain a dual formulation for the optimization problem that satisfies strong duality and which provides, as an addendum to the robustness, linear functionals that act as witnesses of incompatibility with the model implied by Fig. \ref{fig:instrumental_dag}. In comparison with traditional tests for compatibility solely involving the observational data, the functionals obtained here include also a part that acts on the interventional data. The dual formulation is given by
\begin{subequations}
\begin{align}
\max&\quad \sum_{a,x} \tr\left[W_{a,x}\,\asmobs\right] + \sum_a \tr\left[V_a\,\asmint\right] \label{s_eq:dual_objective} \\
\textrm{s.t.}&\quad \sum_x D_\lambda(a|x)\,W_{a,x} + V_a \leq \delta^\xi_{a,\lambda}\,\mathbbm{1},\,\, \sum_a \delta^\xi_{a,\lambda} = 0 \label{s_eq:dual_upper_bound} \\
&\quad \sum_x D_\lambda(a|x)\,W_{a,x} + V_a \geq -\delta^\zeta_{a,\lambda}\,\mathbbm{1},\,\, \sum_a \delta^\zeta_{a,\lambda} = 1 \label{s_eq:dual_lower_bound}
\end{align}
\label{eq:Steering_Dual}
\end{subequations}
Interestingly, causal bounds based on observational data can be derived from such witnesses, as indicated in the theorem below.

\begin{prop}
\label{thm:obs_bnd}
Given witnesses $\{W_{a,x}\}_{a,x}$ and $\{V_a\}_a$ satisfying the constraints of Eq. \eqref{eq:Steering_Dual}, a general inequality for testing the incompatibility of an extended assemblage with the model of the instrumental DAG (Fig. \ref{fig:instrumental_dag}) is given by
\begin{equation}
    \sum_{a,x} \tr[W_{a,x}\,\sigma_{a|x}] > \sum_a \|V^-_a\|_\infty,
\end{equation}
where $V^-_a = (|V_a| - V_a)/2$ are the projections onto the negative subspaces of the interventional part of the witness. A necessary criterion for the usefulness of the observational witness is given by $\sum_x \max_a \|W^+_{a,x}\|_\infty \geq \sum_a \|V^-_a\|_\infty$, where $W^+_{a,x}$ is the projection onto the positive subspace of $W_{a,x}$.
\end{prop}

\begin{proof}
First, it should be noticed that the inequality constraint \eqref{s_eq:dual_upper_bound} implies that no classically correlated assemblage can attain a positive value for the assemblage, since, for any given set of local states $\sigma_{a,\lambda} \geq 0$ with $\tr[\sigma_{a,\lambda}] = f_\lambda$ and $\sum f_\lambda = 1$, the inequality implies
\begin{equation}
\sum_{a,x,\lambda} \tr[W_{a,x}\, D_\lambda(a|x)\,\sigma_{a,\lambda}] + \sum_{a,\lambda} \tr[V_a \sigma_{a,\lambda}] \leq \sum_{a,\lambda} \delta^\xi_{a,\lambda}\, f_\lambda = 0,
\end{equation}
where the last equality results from the constraint $\sum_a \delta^\xi_{a,\lambda} = 0$. The claim follows from noting that $\sum_\lambda \sigma_{a,\lambda}$ is the interventional data associated to the assemblage $\sum_\lambda D_\lambda(a|x)\,\sigma_{a,\lambda}$.

Now, considering only the observational part on the left-hand side, we observe that
\begin{eqnarray}
\nonumber
   & & \sum_{a,x}\tr\left[W_{a,x} \left(\sum_\lambda D_\lambda(a|x)\,\sigma_{a,\lambda}\right)\right] \\ & & \nonumber \leq -\sum_a \tr[(V^+_a - V^-_a) \sum_\lambda \sigma_{a,\lambda}] \nonumber \\
    & &\leq \sum_a \tr[V^-_a\,\sum_\lambda \sigma_{a,\lambda}] \nonumber  \\
    & &\leq \sum_a \|V^-_a\|_\infty\,\left\|\sum_\lambda \sigma_{a,\lambda}\right\|_1 = \sum_a \|V^-_a\|_\infty,
\end{eqnarray}
where the last inequality is obtained by applying the Hölder inequality for Schatten p-norms on the operators $V^-_a$ and $\sum_\lambda \sigma_{a,\lambda}$. For the final equality, it is noted that $\| \sum_\lambda \sigma_{a,\lambda}\|_1 = \tr[\sum_\lambda \sigma_{a,\lambda}] = 1$.
Hence, the first part of the theorem is proved.

For the usefulness criterion, note that, for a generic assemblage $\sigma_{a|x}$, following a similar procedure to the one above, an upper bound for the attainable value of the observational part is given by $\sum_{a,x} \tr[W_{a,x}\sigma_{a|x}] \leq \sum_{a,x} \|W^+_{a,x}\|_\infty \tr[\sigma_{a|x}]$. Since $\tr[\sigma_{a|x}]$ is a conditional distribution $p(a|x)$, it can be used that an extremal distribution would be a deterministic behavior collecting the largest values of $\|W^+_{a,x}\|_\infty$ for each value of $x$, and thus
\begin{equation}
    \sum_{a,x} \tr[W_{a,x}\,\sigma_{a|x}] \leq \sum_x \max_a \|W^+_{a,x}\|_\infty.
\end{equation}
It is clear that the observational part of the witness cannot be useful on its own if the upper bound is no greater than the bound set by the interventional part, $\sum_a \|V^-_a\|_\infty$.
\end{proof}

\subsection{Applications of Interventional Steering}
In what follows, we apply the SDP established in Eq.\ \eqref{eq:Steering_Dual} to scenarios relevant for transmission of quantum information or processing of quantum information. In the bipartite case, we consider the model depicted in Fig. \ref{fig:instrumental_dag} with a quantum node $B$, which can be associated to the task of remote state preparation \cite{bennett2001remote} that can also be seen as the simplest instance of a measurement-based quantum computation \cite{briegel2009measurement,raussendorf2001one}. In its simplest form, assuming knowledge of all measurements and states used, the parties, Alice and Bob, share a maximally entangled two-qubit state, such as $|\Psi^-\rangle = (|01\rangle - |10\rangle)/\sqrt{2}$, and Alice's objective is to project Bob's state onto a given state on the equator of the Bloch's sphere, whose angle is determined by the input $X$, as $\varphi^x$. By communicating a single bit to Bob, Alice fulfills the task by measuring her qubit along the basis $\{\,|\varphi^x_{(a)}\rangle = (|0\rangle + (-1)^a\,e^{i\varphi^x}|1\rangle)/\sqrt{2}\,\}_{a=0,1}$, then informing Bob of the result obtained. The target state is obtained after a $Z$-flip correction on Bob's site, depending on the value of $a$.

The behavior of the robustness for the case of two angles ($|X|=2$) is shown in Fig. \ref{fig:RSP_Measurement}. Assemblages are produced with an underlying model compatible with sharing a singlet state between the parties and Alice applying measurements $\{|\varphi^x_{(a)}\rangle\}$, with fixed $\varphi^0 = 0$ and $\varphi^1$ left as a parameter $\varphi$.

In Fig. \ref{fig:RSP_Schmidt}, the effect of changing the quality of the entanglement shared between the parties is shown. Again, only two different possibilities for the target state are used, but now fixed to $\varphi^0 = 0$ and $\varphi^1 = \pi/2$. The assemblage is now produced by applying measurements on Alice's side on her part of the state $\cos(\theta)|00\rangle + \sin(\theta)|11\rangle$.

\begin{widetext}

\begin{figure}[t!]
\begin{subfigure}{.49\textwidth}
    \centering
    \caption{Varying measurements}
    \includegraphics[width=\textwidth]{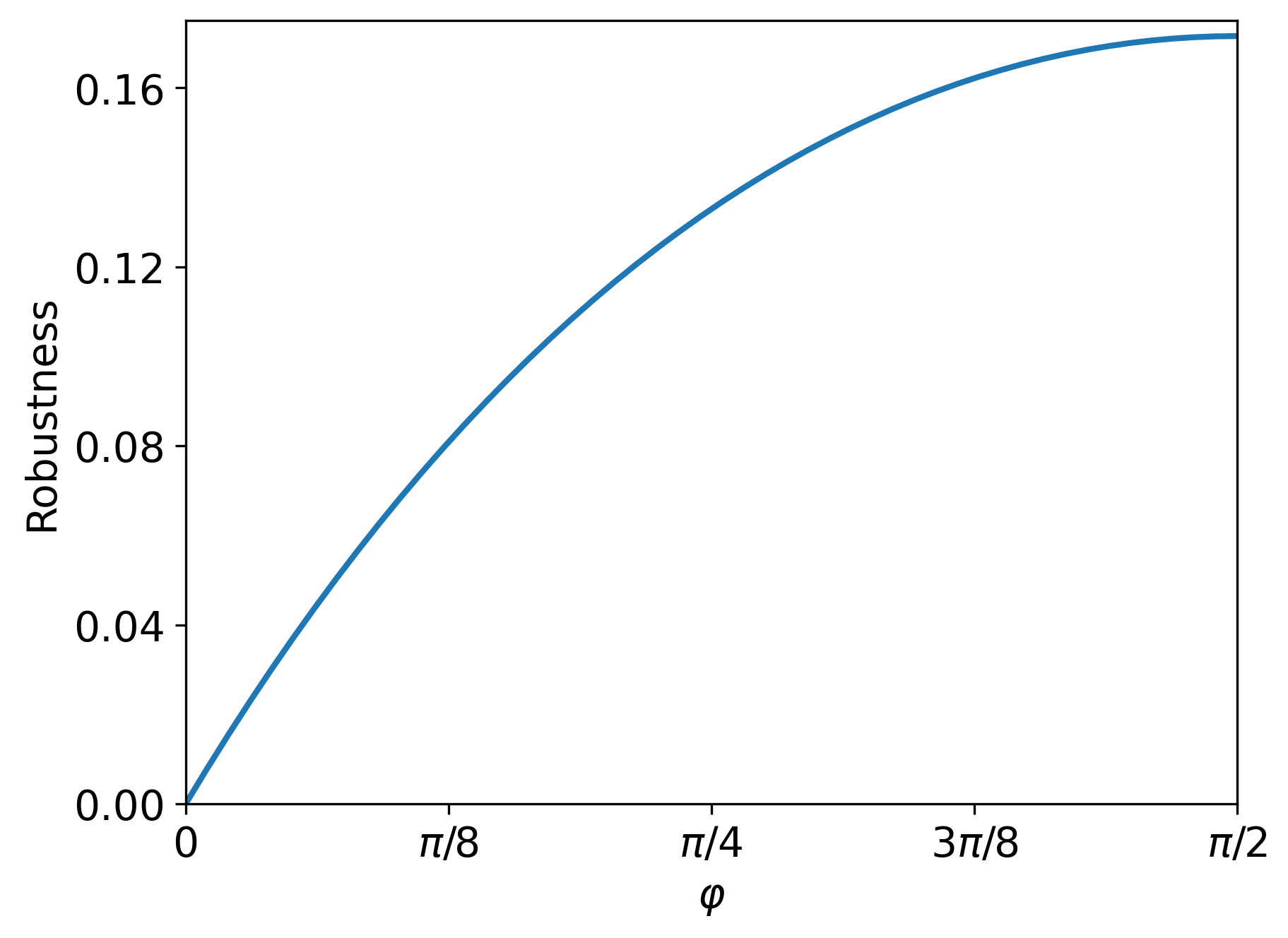}
    \label{fig:RSP_Measurement}
\end{subfigure}
\begin{subfigure}{.49\textwidth}    
    \centering
    \caption{Varying entanglement}
    \includegraphics[width=\textwidth]{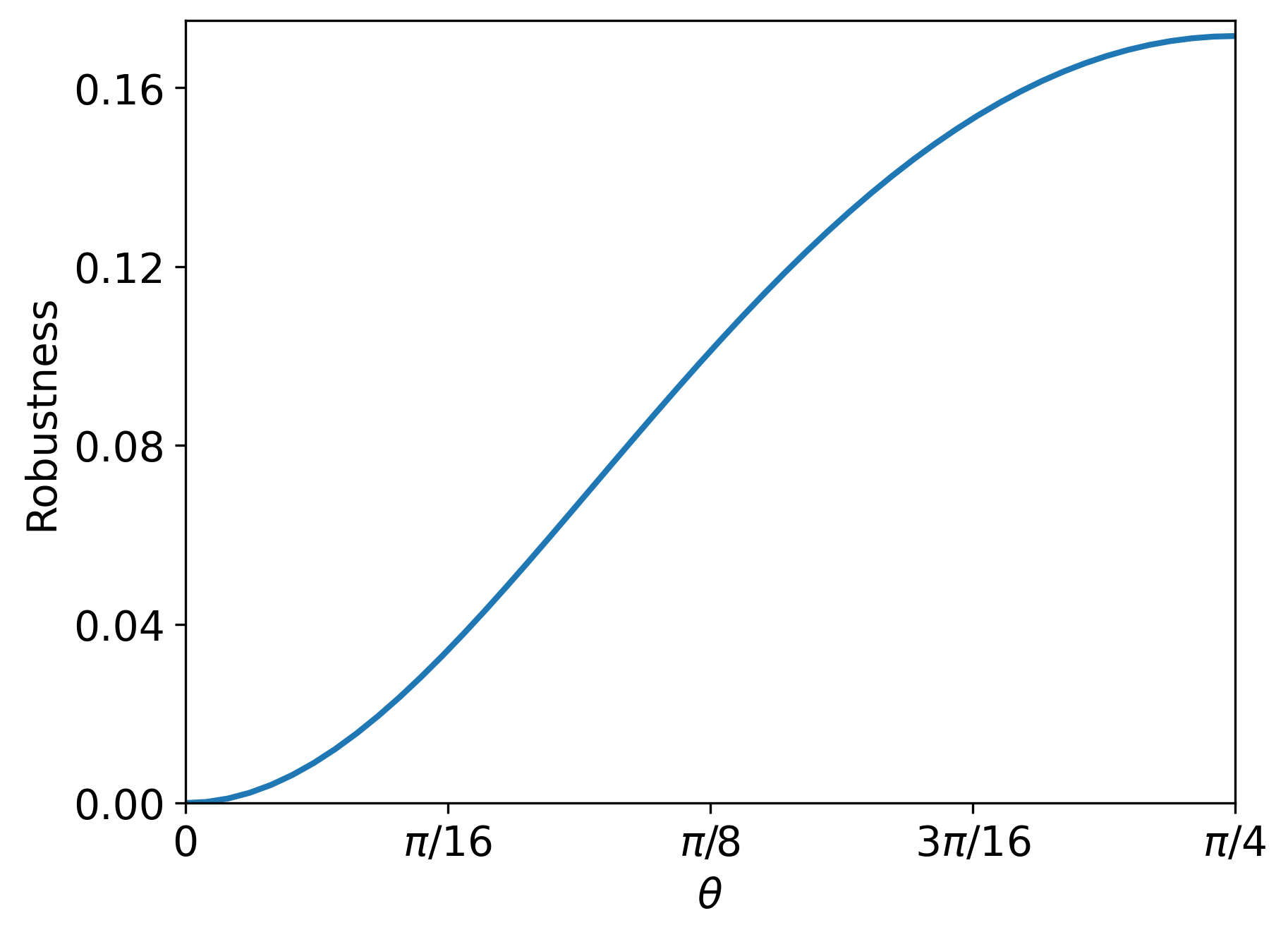}
    \label{fig:RSP_Schmidt}
\end{subfigure}
    \caption{(a) Robustness for remote state preparation with varying angle between $x=0$ and $x=1$ projections. (b) Robustness with fixed measurements for Alice in the equator of the Bloch sphere (projection along $X$, for $x=0$ and along $Y$, for $x=1$) varying the degree of entanglement of the shared state: $\cos(\theta)|00\rangle + \sin(\theta)|11\rangle$.}
    \label{fig:RSP_Curves}
\end{figure}
\end{widetext}

Also in the bipartite scenario, it is known that incompatibilities with the instrumental scenario can be detected for $|X|=3$ using only observational data \cite{nery2018quantum}. The assemblage is obtained with measurements on Alice's side along the eigenstates of $\{-(\sigma_X+\sigma_Z)/\sqrt{2},\, \sigma_X,\,\sigma_Z\}$, where $\sigma_X$ and $\sigma_Z$ are the first and third Pauli matrices, and respectively for $x=0,1,2$, on the state 
\begin{equation}
\rho_v = v\,|\Phi^+\rangle\langle \Phi^+| + (1-v)\,\frac{|00\rangle\langle 00| + |11\rangle\langle 11|}{2}.
\label{eq:state_instrumental_x3}
\end{equation}
It results that, up to numerical precision, critical visibility with interventional data matches the visibility obtained for a standard steering test, without allowing for output communication to Bob, while a higher visibility is required to detect nonclassicality using only observational data and output communication, yet another instance that shows the advantage of including interventions in our description.

Witnesses obtained for this model exemplify an application of Proposition \ref{thm:obs_bnd}, the explicit form of $W_{a,x}$ and $V_a$ for the case of full visibility, $v=1$, are given in Appendix \ref{app:Witnesses}. When applied to the corresponding assemblage, it results in $\tr[W_{a,x}\sigma_{a|x}] \approx 0.672$, while the bound computed from $\sum_a \|V^-_a\|_\infty$ is given by $0.542$.

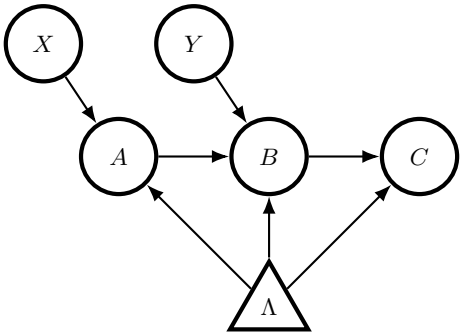
\begin{figure}
\begin{tikzpicture}
    \node[var] (c) at (4,0) {$C$};
    \node[var] (b) at (2,0) {$B$};
    \node[var] (a) at (0,0) {$A$};
    \node[var] (x) at (-1,1.5) {$X$};
    \node[var] (y) at (1,1.5) {$Y$};
    \node[latent] (l) at (2,-2) {$\Lambda$};
    \path[dir] (x) edge (a) (a) edge (b);
    \path[dir] (y) edge (b);
    \path[dir] (b) edge (c); 
    \path[dir] (l) edge (a) (l) edge (b) (l) edge (c);
\end{tikzpicture}
\caption{DAG underlying the adaptative measurement scenario of a one-way quantum computation using a 3-qubits cluster state.}
\label{fig:tripartite_dag}
\end{figure}

Moving to tripartite scenario, consider the DAG of Fig. \ref{fig:tripartite_dag}. Allowing for a cascaded communication structure among the nodes of a multipartite scenario allows for the encoding of an adaptive measurement protocol over shared quantum states, such as the one used for one-way quantum computing \cite{briegel2009measurement,raussendorf2001one}. Compatibility with such a scenario implies, for a quantum output state at $C$ and a shared global state $\rho_\Lambda$, existence of a measurement at the $B$ site that depends on the outcome of the first measurement at $A$'s site and corrections over the final state that depends on the results of previous measurements:
\begin{equation}
    \sigma_{a,b|x,y} = \tr_{A,B}[ M^{(a)}_x \otimes M ^{(b)}_{y,a} \otimes \mathcal{E}_b (\rho_\Lambda)].
\end{equation}

An interventions on $B$ has the effect of removing the indirect link between $A$ and $C$ and separating the nodes conditioned on $\Lambda$,
\begin{equation}
    \sigma_{a,\textrm{do}(b)|x} = \tr_{A,B}[ (M^{(a)}_x \otimes I \otimes \mathcal{E}_b (\rho_\Lambda)].
\end{equation}
A similar test of Eq.\ \eqref{eq:Steering_Dual} can be devised for this scenario, including the interventional data, which allows for the testing of incompatibility with models where only a classical source of correlation $\Lambda$ is available.

Using a standard resource for measurement-based quantum computing, the graph state $|G_3\rangle = (|+0+\rangle + |-1-\rangle)/\sqrt{2}$, where $|\pm\rangle = (|0\rangle \pm |1\rangle)/\sqrt{2}$, and allowing for some noise, quantified by the visibility $v$ of the graph state as $v|G_3\rangle\langle G_3| + (1-v)\mathbbm{1}/8$, the critical value of $v$ for which the state reveals incompatibility with a classical shared resource is found to be approximately $0.577$ when interventional data is included, assuming the assemblage is produced with measurements along the eigenbasis of Pauli matrices $X$ for $x=0$ and $Y$ for $x=1$ for Alice, and along $\cos((-1)^a\,\pi/4)X + \sin((-1)^a\,\pi/4)Y$ and $\cos((-1)^a\,3\pi/4)X + \sin((-1)^a\,3\pi/4)Y$ for Bob (adapting the sign depending on Alice's outcome). This result should be compared with the sole use of observational data, where the critical visibility is found to be approximately $0.744$, a clear reduction in the ability of identifying nonclassicality from the lack of the extra information.

Interestingly, if an extra assumption is made, that the local states of Charlie should not depend individually on the outcomes of Bob or, more strongly, that the experiment is performed in conditions that ensure that Bob has no direct influence on Charlie, the critical visibilities for nonclassical correlations coincide for both using only observational data and including interventional data. This results from the stronger relation that, if no direct communication from Bob to Charlie is ensured, then $\sigma_{a,\textrm{do}(b)|x} = \sum_b \sigma_{a,b|x,y}$, and thus no extra information is carried by the interventional data. In particular, an adapted witness, given by $W_{a,b,x,y} + \sum_b V_{a,b,x}$, is an example of nonclassicality witness for the observational data that is also optimal for its detection for the given assemblage $\sigma_{a,b|x,y}$.

\section{Discussion}
\label{sec:discussion}

Despite its evident ties to classical notions of cause and effect, only recently has Bell's theorem begun to undergo an analysis through the lens of causality theory \cite{wood2015lesson,chaves2015unifying}. This has led to numerous generalizations, especially within the realm of quantum networks \cite{tavakoli2022bell}. However, the exploration of a central concept in causality theory - that of an intervention - remains significantly unexplored.

In this work we give an initial step in this direction, by conducting a detailed analysis of the instrumental scenario \cite{pearl1995testability,chaves2018quantum}, a paradigmatic context where interventions have been introduced within the classical realm. Unlike previous attempts, both in the classical \cite{balke1997bounds,miklin2022causal} and in the quantum contexts \cite{hutter2023quantifying}, that condense all the information coming from interventions in a single quantifier of causality, here we propose to integrate observational and interventional data into a novel form of hybrid Bell inequalities.
 
Focusing on a specific case of the instrumental scenario, where measurement outcomes are binary but the inputs can assume any discrete value, we obtain a complete characterization of the observational-interventional polytope, described by a single class of non-trivial inequality. As demonstrated, this inequality is equivalent to a Hardy-like Bell inequality, albeit describing distinct causal structures and quantum experiments. This equivalence underscores a significant enhancement in the critical detection efficiency required for quantum violations of such inequalities compared to previous methodologies \cite{cao2021detection}. Moving beyond the instrumental scenario, we also proved a general connection between DAGs with interventions and exogeneized DAGs without interventions, establishing a more general mapping explaining the equivalence of our hybrid inequality with a Hardy-like Bell inequality.

Finally, we have also applied the use of interventions a quantum steering scenario \cite{wiseman2007steering,cavalcanti2016quantum}. We show that the use of interventional data allows one to prove the non-classical behaviour of certain experiments that would have a classical simulation if only observational data would be taken into account, thus generalizing to a semi-device-independent scenario the results in Ref. \cite{gachechiladze2020quantifying}.

We believe that interventions are a powerful new tool to understand and witness non-classical behaviour in a variety of causal networks and in their applications to information processing. Their utility extends far beyond the scenarios delineated here and could also be effectively deployed in scenario of growing interest within the quantum community, such as the triangle network \cite{renou2019genuine} or entanglement swapping \cite{branciard2010characterizing,lauand2023witnessing}. Another possibly fruitful application would be to consider the use of interventions on informational principles for quantum theory such as information causality \cite{pawlowski2009information,chaves2015information}, that so far only relies on observational data. We hope our work might trigger further developments in these directions.

\section*{Acknowledgements}
This work was supported by the Serrapilheira Institute (Grant No. Serra-1708-15763), the Simons Foundation (Grant Number 1023171, RC), the Brazilian National Council for Scientific and Technological Development (CNPq) (INCT-IQ and Grant No 307295/2020-6) and the Brazilian agencies MCTIC, CAPES and MEC. PL was supported by São Paulo Research Foundation FAPESP (Grant No. 2022/03792-4). R.N. acknowledges support from the Quantera project Veriqtas.

\bibliography{main_arxiv}

\clearpage 

\appendix

\section{Quantum strategy for the violation of the $I_{22}$ bound}
\label{sec:I22_qstrategy}
In the instrumental scenario $l22$ with $l$ settings and dichotomic measurement, when including the interventional data, the only non trivial bound reduces to
\begin{align}
\nonumber
 I_{l22}=  & &p(b|\Do{a}) -p(a,b|x') + p(a,\bar b|x) \\ \nonumber
     & & + p(\bar a,b|x) - p(\bar a,b|x') \ge 0,
\end{align}

This inequality can be violated, whenever $x \neq x'$ using a bipartite Bell state $\ket{\Phi^+} = (\ket{00} + \ket{11})/\sqrt{2}$ and projective measurements $A^a_x = \ketbra{\alpha_{a,x}}{\alpha_{a,x}}$ and $B^b_a = \ketbra{\beta_{b,a}}{\beta_{b,a}}$ where
\begin{align}
\ket{\alpha_{0,x}} &= \frac{1}{\sqrt{2}} (\ket{0} + e^{i\frac{\pi}{4}} \ket{1}) \\
\ket{\alpha_{0,x'}} &= \frac{1}{\sqrt{2}} (\ket{0} + e^{i\frac{3}{4}\pi} \ket{1}) \\
\ket{\beta_{0,a}} &= \frac{1}{\sqrt{2}} (\ket{0} -i \ket{1}) \\
\ket{\beta_{0,\bar a}} &= \frac{1}{\sqrt{2}} (\ket{0} - \ket{1}) \, .
\end{align}
Using this strategy we reach the negative value $I_{22} = -(\sqrt{2} - 1)/2 \approx -0.2071$, violating maximally the classical bound.
\section{Equivalence between Hardy and CHSH inequalities}
\label{app:Hardy_CHSH}
Now, we consider the Hardy inequalities \eqref{eq: hardy_ineq} and one of its relabelings given by 
\begin{equation}
    \begin{aligned}
        &p(1,0|0,1)+p(0,1|1,0)+p(0,0|0,0)- p(0,0|1,1)\geq0\\
        &p(0,1|0,1)+p(1,0|1,0)+p(1,1|0,0)- p(1,1|1,1)\geq0.\\
    \end{aligned}
\end{equation}
By summing both expressions we have 
\begin{equation}
\begin{aligned}
     D_0 \equiv& p(a\neq b|0,1)+p(a\neq b|1,0)\\
     &+p(a=b|0,0)- p(a=b|1,1)\geq0
\end{aligned}
\end{equation}
 using that $p(a=b|x,y)=1-p(a\neq b|x,y)$ we can write 
 \begin{equation}
 \begin{aligned}
     &-p(a= b|0,1)-p(a= b|1,0)-p(a\neq b|0,0)\\
     &+ p(a\neq b|1,1)+2\geq 0
 \end{aligned}
 \end{equation}
 or simply 
 \begin{equation}
 \begin{aligned}
     D_1\equiv &p(a= b|0,1)+p(a= b|1,0)\\
     &+p(a\neq b|0,0) - p(a\neq b|1,1)\leq 2
 \end{aligned}
 \end{equation}
since $-D_0\leq 0\implies D_1-D_0\leq 2$, i.e. 
 \begin{equation}
 \begin{aligned}
     D_1-D_0&=( p(a= b|0,1)- p(a\neq b|0,1))\\
     &+ (p(a= b|1,0)-p(a\neq b|1,0))\\
     &+( p(a\neq b|0,0)-p(a=b|0,0))\\
     &+ (p(a=b|1,1) - p(a\neq b|1,1))=\\
     &=\langle A_0B_1\rangle+\langle A_1B_0\rangle+\langle A_0B_0\rangle-\langle A_1B_1\rangle\leq 2.
 \end{aligned}
 \end{equation}
 We can do the same procedure for relabelings of the inputs to obtain the inequalities 
 \begin{equation}
     \langle A_{x'}B_y\rangle+\langle A_xB_{y'}\rangle+\langle A_{x'}B_{y'}\rangle-\langle A_xB_y\rangle\leq 2 \quad \forall x\neq x', y\neq y'
 \end{equation}
 Therefore, we have shown Hardy-inequalities $\implies$ CHSH inequalities. Since the CHSH inequalities $\iff p(a,b|x,y)$ is local and $p(a,b|x,y)$ local $\implies$ Hardy inequalities. Then Hardy inequalities $\iff p(a,b|x,y)$ is local \qed.

\section{Explicit form of witnesses for the instrumental scenario with $|X|=3$}
\label{app:Witnesses}

Running SDP \eqref{eq:Steering_Dual} for the assemblage derived from state \eqref{eq:state_instrumental_x3} with three measurement choices for Alice, witnesses $W_{a,x}$ and $V_a$ are obtained numerically as subproduct of the optimization. The explicit form obtained for $W_{a,x}$ when state visibility is $1$ is provided in the table below:
\vspace{2ex}

\begin{tabular}{c|c|c|c|}
\diagbox[height=4ex, width=2em]{\raisebox{0.95\height}{\hphantom{ii.} $x$}}{ \raisebox{-1.6\height}{$a$ \hphantom{l.}}} 
& 0 & 1 & 2 \\ [1ex]
\hline
0
\vphantom{
$\begin{bmatrix}
0 \\ 0 \\ 0
\end{bmatrix}$
}
&
$\begin{bmatrix}
-0.180 & -0.135 \\
-0.135 &  0.0902
\end{bmatrix}$
&
$\begin{bmatrix}
-0.0800 & 0.272 \\
0.272 &  0.0283
\end{bmatrix}$
&
$\begin{bmatrix}
0.268 & -0.0542 \\
-0.0542 & -0.277
\end{bmatrix}$
\\
1
\vphantom{
$\begin{bmatrix}
0 \\ 0 \\ 0
\end{bmatrix}$
}& 
$\begin{bmatrix}
0.0902 &  0.135 \\
0.135  & -0.180
\end{bmatrix}$
&
$\begin{bmatrix}
0.0283 & -0.272 \\
-0.272 & -0.0800
\end{bmatrix}$
&
$\begin{bmatrix}
-0.277 &  0.0542 \\
0.0542 &  0.268
\end{bmatrix}$
\\
\hline
\end{tabular}
\vspace{2ex}

For $V_a$, the resulting form is given as follows:
\begin{eqnarray}
V_0 &= \begin{bmatrix}
-0.104 & -0.0415 \\
-0.0415 & -0.0215
\end{bmatrix} \nonumber\\ [1ex]
V_1 &= \begin{bmatrix}
-0.320 &  \hphantom{-}0.0415 \\
\hphantom{-}0.0415 & -0.403
\end{bmatrix}
\nonumber
\end{eqnarray}

\end{document}